\documentclass[12pt]{article}

\usepackage{setspace}
\usepackage{fullpage}
\usepackage[T1]{fontenc}
\usepackage{ae,aecompl}
\usepackage[dvips]{graphicx}
\usepackage{subfigure}
\usepackage{color}
\usepackage{array}
\usepackage{rotating}
\usepackage{colortbl}
\usepackage{tikz}
\usetikzlibrary{patterns}

\usepackage{amsthm}
\usepackage{amsmath}
\usepackage{amssymb}
\usepackage[utf8]{inputenc}
\usepackage{changepage}
\usepackage[margin=1in, footnotesep=1cm]{geometry}
\usepackage[toc,page]{appendix}
\usepackage[authordate,bibencoding=auto,strict,backend=biber,natbib]{biblatex-chicago}
\usepackage{verbatim}
\usepackage{mathtools}

\definecolor{webgreen}{rgb}{0,0.4,0}
\definecolor{webbrown}{rgb}{0.6,0,0}
\definecolor{purple}{rgb}{0.5,0,0.25}
\definecolor{darkblue}{rgb}{0,0,0.7}
\definecolor{darkred}{rgb}{0.7,0,0}
\definecolor{darkgreen}{rgb}{0,0.7,0}
\usepackage[pdfborder=false,hypertexnames=true]{hyperref}
\hypersetup{colorlinks,citecolor=darkred,filecolor=black,linkcolor=darkblue,urlcolor=webgreen,pdfpagemode=none,
pdfstartview=FitH}

\usepackage{sectsty}
\sectionfont{\centering\normalfont\scshape}
\subsectionfont{\centering\normalfont}

\begin{document}
\begin{spacing}{1.3}

\title{Ex-post implementation with interdependent values \footnote{We are thankful to Debasis Mishra for his invaluable guidance and support. We also thank Arunava Sen and Stephen Morris for insightful comments.}}

\author{Saurav Goyal \footnote{Economics and Planning Unit, Indian Statistical Institute (Delhi)} \hskip 7mm Aroon Narayanan \footnote{Department of Economics, MIT}}

\date{}

\maketitle

\begin{abstract}
    We characterize ex-post implementable allocation rules for single object auctions under quasi-linear preferences with convex interdependent value functions. We show that requiring ex-post implementability is equivalent to requiring that the allocation rule must satisfy a condition that we call {\sl eventual monotonicity (EM)}, which is a weakening of monotonicity, a familiar condition used to characterize dominant strategy implementation.  \\


\end{abstract}

\newpage

\section{Introduction}

We study the single object auction model when agents have interdependent values. In our model, each agent $i$
has a private signal $s_i$ about her value and her (ex-post) value for the object depends on the private signals of all the
$n$ agents: $v_i(s_1,\ldots,s_n)$. The only assumptions we make on the value function $v_i$ is that it is
convex and non-decreasing in the signal $s_i$. Thus, $v_i$ allows for a rich class of
interdependence: the additive model of \citet{BulowKlemperer2002}; the max model of \citet{BBM2019}; and the private values
model in \citet{Myerson1981}. We define a new property for allocation rules
that we call {\sl eventual monotonicity (EM)}. We show that an allocation rule is ex-post implementable (i.e., there exists
a transfer rule such that the allocation rule and the transfer rule form an ex-post incentive compatible mechanism) if and only if
it satisfies eventual monotonicity.

This novel property is a weakening of the
monotonicity property that characterizes dominant strategy implementation \citep{Myerson1981}. Towards understanding this property, let $\ell_{v_i}(s_{-i})$ denote the supremum of all signals of agent $i$ for which $v_i$ takes the same value as it does when $i$ has the lowest possible signal.\footnote{We assume that the signal space for each agent is an interval, so this supremum is also over an interval.} Eventual monotonicity is then the following requirement: for any $s_{-i}$  and any pair of signals $s_i > s'_i$ for agent $i$ with $s_i > \ell_{v_i}(s_{-i})$, the allocation probablity of agent $i$ at $s'_i$ is no less than the allocation probability at $s_i$. This requirement captures two conditions. First, the allocation probability of agent $i$ must be increasing beyond $\ell_{v_i}(s_{-i})$. Two, while it does not need to be increasing below $\ell_{v_i}(s_{-i})$, it must be the case that the allocation probability for any type above $\ell_{v_i}(s_{-i})$ is more than the allocation probability for any type below $\ell_{v_i}(s_{-i})$. See Figure \ref{fig:END} below for an illustration.

Extending Myerson's monotonicity characterization to interdependent values models may be useful in solving optimal auction problems in these models with ex-post incentive compatibility as the solution concept. The interdependent values model is considered to be a model that reflects many realistic problems where agents have imperfect information about values: \citet{MilgromWeber,BulowKlemperer2002,BBM2019} contain several motivating examples.\footnote{\citet{BBM2019} describes a setting where an agent's signal denotes her private use. If an agent comes to know of various uses of the object, she makes the best use of it.
This is captured in their max model, a common value model, given by $v_i(s_1,\ldots,s_n)=\max s_j$. Our results cover their model.}

Our characterization and a payment equivalence formula provide a handy method of deriving an expression for the revenue of the revenue-maximizing mechanism in terms of maximizing an objective function within the class of eventually monotone allocation rules. The tractability of deriving the exact optimal EPIC mechanism in general will depend on the form of the value functions. We show that for strictly increasing value functions, the usual Myersonian approach works. We are also able to show that the BIC and IIR inclusive posted price mechanism studied in \citet{BulowKlemperer2002}, and shown to be optimal in \citet{BBM2019}, can in fact be achieved in expectation by an EPIC and EPIR mechanism using the characterization that we develop, under the condition that the object must be sold.

{\bf Relation to the literature.} Much of the literature on interdependent values model focuses on Bayesian incentive compatibility (BIC) and interim individual rationality (IIR). However, BIC has been criticized for its reliance on the specific priors of the agents, and hence raises questions of robustness. We use the stronger and more robust concepts of \textit{ex-post} incentive compatibility (EPIC) and \textit{ex-post} individual rationality (EPIR) instead. While many interesting but specific models of interdependent values have been studied, a general characterization of ex-post implementability paralleling that for private values has proved elusive, partly because of the difficulty in identifying the set of binding constraints. \citet{Maskin1992} shows that the Vickrey auction may not be ex-post implementable without the single-crossing property for the value functions \footnote{The single-crossing property ensures that the agent with the highest signal is also the agent with the highest value, which ensures that efficient auctions produce monotonic, and hence implementable, allocation rules. Our result is orthogonal - we show that under convexity there is an exact characterization of implementable rules.}. \citet{jehiel2006} demonstrates that with multidimensional types and interdependent values, there are no non-trivial (non-constant) ex-post implementable deterministic social choice functions, while \citet{Bikhchandani2006} shows that this is not the case for private values. We fill this gap in the literature by providing a characterization of ex-post implementability in models of interdependent values.

The paper closest to ours is \citet{ULKU2014}, which provides sufficient conditions under which allocation rules are ex-post implementable. Building on a result from \citet{CremMclean}, they show that under increasing differences across signals, monotonicity of the allocation rule in own signal is sufficient for ex-post implementation. However they do not provide a characterization.

The rest of the paper proceeds as follows. Section 2 specifies the setting of the problem. Section 3 builds the characterization. Section 4 explores the problem of revenue maximization, both in general and in specific models.

\section{Model}

\newtheorem{definition}{Definition}[]
\newtheorem{question}{Question}[]
\newtheorem{observation}{Observation}[]
\newtheorem{conjecture}{Conjecture}[]
\newtheorem{claim}{Claim}[]
\newtheorem{lemma}{Lemma}[]
\newtheorem{proposition}{Proposition}[]
\newtheorem{corollary}{Corollary}[]
\newtheorem{theorem}{Theorem}[]

\newcounter{example}[section]
\newenvironment{example}[1][]{\refstepcounter{example}\par\medskip
   \noindent \textbf{Example~\theexample. #1} \rmfamily}{\medskip}

Let $N = \{1,2,...,n\}$ be the set of agents. There is one good to be allocated. The agents have quasi-linear interdependent preferences for the good. Each agent receives signals from $S_i = [\underline{s_i},\overline{s_i}]$. Let $v_i:S^{n} \to \mathbb{R}$ be the valuation function of each agent $i$. $v_i(s_i,s_{-i})\equiv v_i(s_1,s_2,...,s_n)$ denotes the valuation of agent $i$ at signal profile $(s_1,s_2,...,s_n)$. We assume that $v_i(s_i,s_{-i})$ is convex, continuous and non-decreasing in $s_i$ on $[\underline{s_i},\overline{s_i}]$,
for all $i$. For every $i \in N$ and for every $s_{-i}$, define
\begin{align*}
\ell_{v_i}(s_{-i}) &:= \sup\{s_i \in S_i: v_i(s_i,s_{-i})=v_i(\underline{s}_i,s_{-i})\}
\end{align*}
Since $v_i$ is convex, continuous, and non-decreasing, $\ell_{v_i}(s_{-i})$ is the supremum of the interval where $v_i$ is constant. For instance, in the max value function studied in \citet{BBM2019}, $v_i(s_i,s_{-i})=\max_j s_j$
and $\ell_{v_i}(s_{-i})=\max_{j \ne i}s_j$.

We invoke the revelation principle to focus only on direct mechanisms. Thus a mechanism in this setting is the tuple $(q,p)$ where $q_i : S^n \to [0,1]$ is the allocation rule for agent $i$ and $p_i: S^n \to \mathbb{R}$ is the payment rule for agent $i$.

\begin{definition}
An allocation rule $q$ is \textbf{ex-post implementable} if there exists a payment rule $p$ such that the following holds $\forall i$ $\forall s_{-i}$:
$$u_i(s_i,s_{-i})\geq u_i(s_i',s_{-i})+[v_i(s_i,s_{-i})-v_i(s_i',s_{-i})]q_i(s_i',s_{-i}), \ \forall \ s_i,s_i' $$

where $u_i(s_i,s_{-i})=q_i(s_i,s_{-i})v_i(s_i,s_{-i})-p_i(s_i,s_{-i})$ is the net payoff function for agent $i$.
In this case, we say that the mechanism $(q,p)$ is \textbf{ex-post incentive compatible (EPIC)}.
\end{definition}

\section{Eventual monotonicity and ex-post implementation}
Our main result characterizes ex-post implementable allocation rules. The approach we take to derive this result is fairly standard - we derive a necessary condition and then show that it also turns out to be sufficient. Here the appropriate necessary condition is a suitable modification of monotonicity.

\begin{definition}

The allocation rule $q$ is \textbf{eventually monotone} if $\forall i$, $\forall s_{-i}$ and for every $s_i > s'_i$ with $s_i > \ell_{v_i}(s_{-i})$, we have $$q_i(s_i,s_{-i}) \ge q_i(s'_i,s_{-i})$$

\end{definition}

Figure \ref{fig:END} illustrates what an EM rule might look like.

\begin{figure}[htp]
    \centering
    \includegraphics[width=8cm]{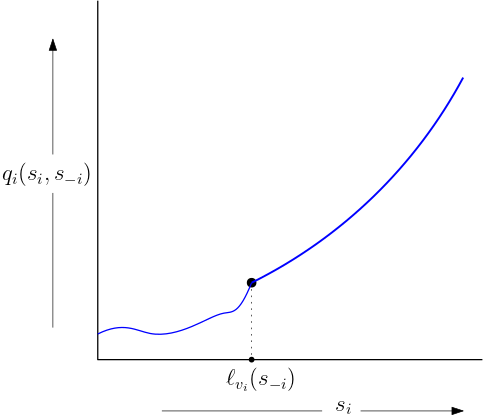}
    \caption{An eventually monotone rule}
    \label{fig:END}
\end{figure}

Now we are ready to state and prove our main result.

\begin{theorem}\label{epicthm}
An allocation rule $q$ is ex-post implementable if and only if it is eventually monotone.
\end{theorem}

It is direct to note that EPIC is equivalent to dominant strategy incentive compatibility (DSIC) under private values, and that eventual monotonicity is just monotonocity. Hence, we obtain the standard characterization of dominant-strategy implementable rules as a corollary of the above theorem.

\section{Optimal mechanisms}
In this section, we assume that the signals of the agents are independent and identically distributed, each drawn from some distribution $F$ with full support on $S$. We will first derive a condition for revenue maximization, and then explore its implications for some value functions.

Building on Theorem \ref{epicthm}, we get the following characterization of EPIC mechanisms. 

\begin{proposition}\label{epicprop}
A mechanism $(q,p)$ is ex-post incentive compatible if and only if $q$ is eventually monotone and $p$ is given by \footnote{Note that $v_i$ need not be differentiable with respect to $s_i$ everywhere, but
convexity ensures that it always has a sub-gradient, which is equal to its derivative whenever the latter exists (which is almost everywhere). We write $\frac{\partial v_i(x,s_{-i})}{\partial x}$
to denote the derivative of $v_i(\cdot,s_{-i})$ with respect to the first argument at $x$ whenever
it is differentiable. Else, it is some selection from the subgradient of $v_i(\cdot,s_{-i})$ at $x$.}:
    \begin{align*}
        p_i(s_i,s_{-i}) = \ &p_i(\underline{s},s_{-i})-q_i(\underline{s},s_{-i})v_i(\underline{s},s_{-i}) \\
        &+ q_i(s_i,s_{-i})v_i(s_i,s_{-i})-\int^{s_i}_{\underline{s}}\frac{\partial v_i(x,s_{-i})}{\partial x} q_i(x,s_{-i}) \ dx
    \end{align*}
\end{proposition}

The following generalization of the private values version of virtual value is well-known in the literature.
\begin{definition}
The \textbf{virtual value} of an agent at a signal profile is given by
$$J_i (s_i,s_{-i}) = v_i(s_i,s_{-i})-\frac{1-F(s_i)}{f(s_i)} \frac{\partial v_i(s_i,s_{-i})}{\partial s_i}$$
\end{definition}

Also, we consider rules that ensure that the agents continue to want to participate ex-post.

\begin{definition}
A mechanism is \textbf{ex-post individually rational (EPIR)} if the following holds $\forall i$ $\forall s_i$ $\forall s_{-i}$:
$$u_i(s_i,s_{-i})\geq 0 $$
\end{definition}

We say that a mechanism $(q,p)$ is \textbf{optimal} if it is EPIC and EPIR and generates more expected
revenue than any other EPIC and EPIR mechanism. We then have the following formula for the revenue in an optimal mechanism.

\begin{lemma}\label{revform}
The optimal mechanism maximizes
$$
\int_{s \in S^n} \Bigl[ \sum_{i \in N} J_i (s) q_i(s) \Bigr] \ f(s)\ ds
$$
within the class of eventually monotone allocation rules.
\end{lemma}

\subsection{Strictly increasing value functions}

When the value functions are strictly increasing in each agent's own signal, then it is direct to observe that EM reduces to monotonicity. In this case the optimization problem identified in Lemma \ref{revform} can be solved by the standard ironing technique of \citet{Myerson1981}. We state this observation in the following proposition:

\begin{proposition}
Suppose each agent's value function is strictly increasing in her own signal. Then, the optimal mechanism allocates the object to the agent with the highest ironed virtual value, and charges the payments that ex-post implements this allocation rule.
\end{proposition}

Note also that the case of private values is covered by this result.

\subsection{Additive signals}

Using the framework developed above, we derive a result for the additive signals model, which is closely related to the one studied in \citet{BulowKlemperer2002}. In this model, the value functions are given by $v_i(s_1,s_2,...,s_n) = c_1 s_1 + c_2 s_2 +...+ c_n s_n$ with $c_i \geq 0$ for all $i$. An example of such a model is the simple additive value case - where each agent's value for the object is the simple sum of the signals of each agent. Denote by $c_i \frac{1-F(s_i)}{f(s_i)}$ the adjusted hazard rate of agent $i$.

\begin{proposition}\label{addthm}
 Suppose the distribution $F$ of the signals satisfies the monotone hazard rate property. Then a mechanism is optimal if and only if at every signal profile it allocates the object with probability $1$ to the agent with the lowest adjusted hazard rate if that agent's virtual value is non-negative and does not allocate the object otherwise.
\end{proposition}

An interesting corollary of this theorem is the following: when the weights $c_i$ are the same for every agent (i.e. the simple additive values case) and the distribution of signals has a monotone hazard rate, the mechanism allocates the object to the agent with the highest signal. This is because the agent with the lowest adjusted hazard rate at any profile must also have the highest signal, by the monotone hazard rate assumption.

\subsection{Max function}

\citet{BBM2019} derive the optimal Bayesian incentive compatible (BIC) and interim individually rational (IIR) mechanism for the maximum value function, $v_i(s_1,s_2,...,s_n)= \max(s_1,s_2,...,s_n)$. We use our framework to identify the optimal mechanism for this form of interdependence, under the assumption that the object must be sold.

\begin{proposition}\label{maxthm}
When it is mandatory to sell the object, mechanisms that allocate the object with probability $c_i$ to agent $i$ at all signal profiles, for some $c_i$ such that $\sum_{i}c_i = 1$, are optimal.
\end{proposition}

Each of these optimal mechanisms guarantees the same revenue in expectation as the BIC and IIR equal allocation rule of \citet{BBM2019}, which allocates with probability $\frac{1}{n}$ to each agent at every profile. Since their optimal mechanism is not ex-post IR, our result identifies an ex-post IR implementation of their mechanism.

\newpage

\printbibliography

\begin{appendices}

\section{Proofs}

\begin{proof}[\textbf{Proof of Theorem \ref{epicthm}}]

Suppose $q$ is ex-post implementable. Then, there exists a payment rule $p$ such that $(q,p)$ is EPIC. Let $s_i'' \geq s_i'$. EPIC implies
$$q_i(s_i'',s_{-i})v_i(s_i'',s_{-i})-p_i(s_i'',s_{-i}) \geq q_i(s_i',s_{-i})v_i(s_i'',s_{-i})-p_i(s_i',s_{-i})$$

and
$$q_i(s_i',s_{-i})v_i(s_i',s_{-i})-p_i(s_i',s_{-i}) \geq q_i(s_i'',s_{-i})v_i(s_i',s_{-i})-p_i(s_i'',s_{-i})$$

Adding these two together,
\begin{equation}\label{6}
    (q_i(s_i'',s_{-i})-q_i(s_i',s_{-i}))(v_i(s_i'',s_{-i})-v_i(s_i',s_{-i}))\geq 0
\end{equation}

Let $s_i'' > s_i' > \ell_{v_i}(s_{-i}) \geq s_i$. Then $v_i(s_i'',s_{-i}) > v_i(s_i',s_{-i}) > v_i(s_i,s_{-i})$. Hence by \ref{6}, we must have $$q_i(s_i'',s_{-i}) \geq q_i(s_i',s_{-i}) \geq q_i(s_i,s_{-i})$$

Thus, $q$ must be eventually monotone. 

For the converse, suppose $q$ is eventually monotone. Let $p$ be given by 

$$p_i(s_i,s_{-i})=p_i(\underline{s},s_{-i})+q_i(s_i,s_{-i})v_i(s_i,s_{-i})-q_i(\underline{s},s_{-i})v_i(\underline{s},s_{-i})-\int^{s_i}_{\underline{s}}\frac{\partial v_i(x,s_{-i})}{\partial x} q_i(x,s_{-i}) \ dx$$

Fix $i \in N, \ s_{-i}$, and consider $s_i, \ s_i'\in S$. We need to prove that the following quantity is non-negative.
$$ \rho \coloneqq u_i(s_i,s_{-i}) - u_i(s_i',s_{-i}) - [v_i(s_i,s_{-i})-v_i(s_i',s_{-i})]q_i(s_i',s_{-i})$$
$$=(q_i(s_i,s_{-i})v_i(s_i,s_{-i})-p_i(s_i,s_{-i}))-(q_i(s_i',s_{-i})v_i(s_i,s_{-i})-p_i(s_i',s_{-i})) $$
$$=\int_{\underline{s}}^{s_i}\frac{\partial v_i(x,s_{-i})}{\partial x} q_i(x,s_{-i}) \ dx-\int_{\underline{s}}^{s_i'}\frac{\partial v_i(x,s_{-i})}{\partial x} q_i(x,s_{-i}) \ dx-q_i(s_i',s_{-i})[v_i(s_i,s_{-i})-v_i(s_i',s_{-i})]$$

Let us proceed by taking cases.

Case: $s_i \geq s_i'> \ell_{v_i}(s_{-i})$
$$\rho=\int_{s_i'}^{s_i}\frac{\partial v_i(x,s_{-i})}{\partial x} q_i(x,s_{-i}) \ dx-q_i(s_i',s_{-i})[v_i(s_i,s_{-i})-v_i(s_i',s_{-i})]$$
$$\geq \int_{s_i'}^{s_i}\frac{\partial v_i(x,s_{-i})}{\partial x} q_i(s_i',s_{-i}) \ dx-q_i(s_i',s_{-i})[v_i(s_i,s_{-i})-v_i(s_i',s_{-i})]$$
$$= q_i(s_i',s_{-i}) \int_{s_i'}^{s_i} dv_i(x,s_{-i}) -q_i(s_i',s_{-i})[v_i(s_i,s_{-i})-v_i(s_i',s_{-i})]=0$$

Case: $s_i > \ell_{v_i}(s_{-i}) \geq s_i'$
$$\rho = \int_{\ell_{v_i}(s_{-i})}^{s_i}\frac{\partial v_i(x,s_{-i})}{\partial x} q_i(x,s_{-i}) \ dx-q_i(s_i',s_{-i})[v_i(s_i,s_{-i})-v_i(s_i',s_{-i})]$$
$$\geq \int_{\ell_{v_i}(s_{-i})}^{s_i}\frac{\partial v_i(x,s_{-i})}{\partial x} q_i(s_i',s_{-i}) \ dx-q_i(s_i',s_{-i})[v_i(s_i,s_{-i})-v_i(s_i',s_{-i})]$$
$$= q_i(s_i',s_{-i}) \int_{\ell_{v_i}(s_{-i})}^{s_i} dv_i(x,s_{-i}) -q_i(s_i',s_{-i})[v_i(s_i,s_{-i})-v_i(s_i',s_{-i})]$$
$$= q_i(s_i',s_{-i}) [v_i(s_i,s_{-i})-v_i(\ell_{v_i}(s_{-i}),s_{-i})] -q_i(s_i',s_{-i})[v_i(s_i,s_{-i})-v_i(s_i',s_{-i})]$$
$$= q_i(s_i',s_{-i}) [v_i(s_i,s_{-i})-v_i(s_i',s_{-i})] -q_i(s_i',s_{-i})[v_i(s_i,s_{-i})-v_i(s_i',s_{-i})]=0$$

Case: $\ell_{v_i}(s_{-i}) > s_i  \geq s_i'$
$$\rho = -q_i(s_i',s_{-i})[v_i(s_i,s_{-i})-v_i(s_i',s_{-i})]=0$$

Case: $s_i < s_i'<\ell_{v_i}(s_{-i})$
$$\rho = 0$$

Case: $s_i < \ell_{v_i}(s_{-i}) < s_i'$

$$\rho = -\int_{\ell_{v_i}(s_{-i})}^{s_i'}\frac{\partial v_i(x,s_{-i})}{\partial x} q_i(x,s_{-i}) \ dx-q_i(s_i',s_{-i})[v_i(s_i,s_{-i})-v_i(s_i',s_{-i})] $$
$$\geq -q_i(s_i',s_{-i}) \int_{\ell_{v_i}(s_{-i})}^{s_i'}\frac{\partial v_i(x,s_{-i})}{\partial x}  \ dx-q_i(s_i',s_{-i})[v_i(s_i,s_{-i})-v_i(s_i',s_{-i})] $$
$$= -q_i(s_i',s_{-i})[ v_i(s_i',s_{-i})-v_i(\ell_{v_i}(s_{-i}),s_{-i}) ] -q_i(s_i',s_{-i})[v_i(s_i,s_{-i})-v_i(s_i',s_{-i})] $$
$$= -q_i(s_i',s_{-i})[ v_i(s_i',s_{-i})-v_i(s_i,s_{-i}) ] -q_i(s_i',s_{-i})[v_i(s_i,s_{-i})-v_i(s_i',s_{-i})]=0 $$

Case: $\ell_{v_i}(s_{-i})<s_i< s_i'$
$$\rho = -\int_{s_i}^{s_i'}\frac{\partial v_i(x,s_{-i})}{\partial x} q_i(x,s_{-i}) \ dx-q_i(s_i',s_{-i})[v_i(s_i,s_{-i})-v_i(s_i',s_{-i})] $$
$$\geq -q_i(s_i',s_{-i}) \int_{s_i}^{s_i'}\frac{\partial v_i(x,s_{-i})}{\partial x}  \ dx-q_i(s_i',s_{-i})[v_i(s_i,s_{-i})-v_i(s_i',s_{-i})] $$
$$= -q_i(s_i',s_{-i})[ v_i(s_i',s_{-i})-v_i(s_i,s_{-i}) ] -q_i(s_i',s_{-i})[v_i(s_i,s_{-i})-v_i(s_i',s_{-i})]=0 $$

Thus we have ex post incentive compatibility of $(q,p)$, and hence $q$ is ex-post implementable.
\end{proof}

\begin{proof}[\textbf{Proof of Proposition \ref{epicprop}}]
Let $(q,p)$ be an EPIC mechanism. We already know from Theorem \ref{epicthm} that $q$ must be eventually monotone. So all we need to show is that the payments take the form specified in the hypothesis. 

Let $u_i$ denote the utility of agent $i$ from $(q,p)$. Since the value functions are convex, it is reasonable to expect that the net payoff functions $u_i$ are convex. This is indeed the case.

\begin{claim}\label{convexu}
For each $i$, $u_i$ is convex in $s_i$.
\end{claim}

\begin{proof}[Proof of Claim \ref{convexu}]
 Consider any $s_1,s_3 \in S$. Let $s_2=\lambda s_1+(1-\lambda) s_3$, where $\lambda \in (0,1)$.
By convexity of $S, \ s_2 \in S$. By incentive compatibility,
\begin{equation}\label{1}
    u_i(s_1,s_{-i})\geq u_i(s_2,s_{-i})+[v_i(s_1,s_{-i})-v_i(s_2,s_{-i})]q_i(s_2,s_{-i})
\end{equation}
\begin{equation}\label{2}
    u_i(s_3,s_{-i})\geq u_i(s_2,s_{-i})+[v_i(s_3,s_{-i})-v_i(s_2,s_{-i})]q_i(s_2,s_{-i})
\end{equation}

Multiplying \ref{1} with $\lambda$ and \ref{2} with $1-\lambda$ and summing
\begin{align*}
    \lambda u_i(s_1,s_{-i}) & + (1-\lambda)u_i(s_3,s_{-i}) \\
    &\geq u_i(s_2,s_{-i})+q_i(s_2,s_{-i})[\lambda v_i(s_1,s_{-i})+(1-\lambda)v_i(s_3,s_{-i})-v_i(s_2,s_{-i})] \\
    &\geq u_i(s_2,s_{-i})
\end{align*}

where the last inequality follows from convexity of $v_i$.
\end{proof}

The next step of the proof uses the convexity of $u_i$ to establish the relationship between $u_i$ and $v_i$.

\begin{claim}\label{gradu}
Given ex-post incentive compatibility, $\frac{\partial u_i}{\partial s_i}=\frac{\partial v_i}{\partial s_i}q_i(s_i,s_{-i})$ almost everywhere.
\end{claim}

\begin{proof}[Proof of Claim \ref{gradu}]
By convexity of $u_i$ and $v_i$, they are not differentiable at most in a pair of measure $0$ subsets. Since their union will continue to be measure $0$, we must have that $u_i$ and $v_i$ are together differentiable almost everywhere.

Let $s_i$ be an interior point where $u_i$ and $v_i$ are both differentiable and let $\epsilon > 0$. By incentive compatibility,
\begin{equation}\label{3}
    \frac{u_i(s_i+\epsilon,s_{-i})-u_i(s_i,s_{-i})}{\epsilon}\geq \frac{v_i(s_i+\epsilon,s_{-i})-v_i(s_i,s_{-i})}{\epsilon}q_i(s_i,s_{-i})
\end{equation}

\begin{equation}\label{4}
    \frac{u_i(s_i,s_{-i})-u_i(s_i+\epsilon,s_{-i})}{\epsilon}\geq \frac{v_i(s_i,s_{-i})-v_i(s_i+\epsilon,s_{-i})}{\epsilon}q_i(s_i+\epsilon,s_{-i})
\end{equation}

Rewriting \ref{4}, we get,
\begin{equation}\label{5}
    \frac{v_i(s_i+\epsilon,s_{-i})-v_i(s_i,s_{-i})}{\epsilon}q_i(s_i+\epsilon,s_{-i}) \geq \frac{u_i(s_i+\epsilon,s_{-i})-u_i(s_i,s_{-i})}{\epsilon}
\end{equation}

Let $\epsilon \rightarrow 0$ and use the sandwich theorem on \ref{3} and \ref{5} to get,

\begin{equation*}
    \frac{\partial u_i}{\partial s_i}=\frac{\partial v_i}{\partial s_i}q_i(s_i,s_{-i})
\end{equation*}
\end{proof}

We  use  the  fundamental  theorem  of  calculus  to  derive  what $p$ must look like.

\begin{claim}\label{paycorr}
The payment rule $p$ must be of the form
$$p_i(s_i,s_{-i})=p_i(\underline{s},s_{-i})+q_i(s_i,s_{-i})v_i(s_i,s_{-i})-q_i(\underline{s},s_{-i})v_i(\underline{s},s_{-i})-\int^{s_i}_{\underline{s}}\frac{\partial v_i(x,s_{-i})}{\partial x} q_i(x,s_{-i}) \ dx$$
\end{claim}

\begin{proof}[Proof of Claim \ref{paycorr}]
From Claim \ref{gradu} and using the fundamental theorem of calculus\footnote{Specifically, the fundamental theorem of Lebesgue integral calculus for absolutely continuous functions.},
$$u_i(s_i,s_{-i})=u_i(\underline{s},s_{-i})+\int_{\underline{s}}^{s_i}\frac{\partial v_i}{\partial x}q_i(x,s_{-i}) \ dx$$
which gives us the required equality.
\end{proof}

For the converse, suppose $q$ is eventually monotone and the payments formula take the form specified. The proof of the converse of Theorem \ref{epicthm} shows that $(q,p)$ is EPIC, and hence we are done.

\end{proof}

\begin{proof}[\textbf{Proof of Lemma \ref{revform}}]

As is standard in the case of private values, in this setting as well IR reduces to ensuring IR of the lowest type.
\begin{claim}\label{IRlow}
Given ex-post incentive compatibility, a mechanism is individually rational if and only if $p_i(\underline{s},s_{-i})\leq q_i(\underline{s},s_{-i})v_i(\underline{s},s_{-i})$
\end{claim}

\begin{proof}[Proof of Claim \ref{IRlow}]
This is a consequence of the following:
$$u_i(s_i,s_{-i}) = -p_i(\underline{s},s_{-i})+q_i(\underline{s},s_{-i})v_i(\underline{s},s_{-i})+\int^{s_i}_{\underline{s}}\frac{\partial v_i(x,s_{-i})}{\partial x} q_i(x,s_{-i}) \ dx \geq 0 \ \forall \ s_i$$

$$\Leftrightarrow$$
$$p_i(\underline{s},s_{-i})\leq q_i(\underline{s},s_{-i})v_i(\underline{s},s_{-i})$$
\end{proof}

Using Proposition \ref{epicprop}, the expression for expected revenue in any EPIC mechanism is
$$\sum_{i \in N} \int_{s \in S}p_i(s)f(s) \ ds $$
$$=\sum_{i \in N} \int_{s \in S}\left \{p_i(\underline{s},s_{-i})+q_i(s_i,s_{-i})v_i(s_i,s_{-i})-q_i(\underline{s},s_{-i})v_i(\underline{s},s_{-i})-\int^{s_i}_{\underline{s}}\frac{\partial v_i(x,s_{-i})}{\partial x} q_i(x,s_{-i}) \ dx \right \} f(s) \ ds$$

\begin{eqnarray*}
=\sum_{i \in N} \int_{s_{-i} \in S_{-i}} \int_{s_i \in S_i} [ p_i(\underline{s},s_{-i})+q_i(s_i,s_{-i})v_i(s_i,s_{-i})-q_i(\underline{s},s_{-i})v_i(\underline{s},s_{-i})-  \\
 \int^{s_i}_{\underline{s}}\frac{\partial v_i(x,s_{-i})}{\partial x} q_i(x,s_{-i}) \ dx ] f(s_i)  \ ds_i \ f(s_{-i})\ ds_{-i}
\end{eqnarray*}

\begin{align*}
= &\sum_{i \in N} \int_{s_{-i} \in S_{-i}} \Bigl\{ \int_{\underline{s}}^{\overline{s}} q_i(s_i,s_{-i})v_i(s_i,s_{-i})f(s_i) \ ds_i \\
& -\int_{\underline{s}}^{\overline{s}} \Bigl[ \int_{\underline{s}}^{s_i}\frac{\partial v_i(x,s_{-i})}{\partial x} q_i(x,s_{-i}) \ dx \Bigr] f(s_i) \ ds_i \Bigr\} \ f(s_{-i})\ ds_{-i} \\
& + \ \sum_{i \in N} \int_{s_{-i} \in S_{-i}} [p_i(\underline{s},s_{-i})-q_i(\underline{s},s_{-i})v_i(\underline{s},s_{-i})] f(s_{-i}) \ ds_{-i}
\end{align*}

Changing the order of integration, we can write,
$$ \int_{\underline{s}}^{\overline{s}} q_i(s_i,s_{-i})v_i(s_i,s_{-i})f(s_i) \ ds_i-\int_{\underline{s}}^{\overline{s}} \left \{ \int_{\underline{s}}^{s_i}\frac{\partial v_i(x,s_{-i})}{\partial x} q_i(x,s_{-i}) \ dx \right \} f(s_i) \ ds_i $$

$$=\int_{\underline{s}}^{\overline{s}} \left \{ v_i(s_i,s_{-i})-\frac{1-F(s_i)}{f(s_i)} \frac{\partial v_i(s_i,s_{-i})}{\partial s_i} \right \} f(s_i)q_i(s_i,s_{-i}) \ ds_i$$

Thus, expected revenue is
$$
\sum_{i \in N} \int_{s_{-i} \in S_{-i}} \left \{ \int_{s_i \in S_i} [J_i (s_i,s_{-i}) f(s_i)q_i(s_i,s_{-i})] \ ds_i \right \} \ f(s_{-i})\ ds_{-i} $$
$$ + \ \sum_{i \in N} \int_{s_{-i} \in S_{-i}} [p_i(\underline{s},s_{-i})-q_i(\underline{s},s_{-i})v_i(\underline{s},s_{-i})] f(s_{-i}) \ ds_{-i}
$$

The second term in the above expression is non-positive under EPIR. Thus, in order to maximize revenue we must set $p_i(\underline{s},s_{-i}) = q_i(\underline{s},s_{-i})v_i(\underline{s},s_{-i})$. As a result, the optimal mechanism maximizes the first term, which reduces to:

$$
\int_{s \in S^n} \Bigl[ \sum_{i \in N} J_i (s) q_i(s) \Bigr] \ f(s)\ ds
$$

\end{proof}

\begin{proof}[\textbf{Proof of Proposition \ref{addthm}}]
If we were to pointwise maximize the revenue, we would not allocate the object at signal profiles where the virtual value of each agent is negative. At other profiles, note that the virtual value of agent $i$ is $J_i(s) = [\sum_{j \in N} c_j s_j - c_i \frac{1-F(s_i)}{f(s_i)}]$. Thus,

\begin{align*}
    & \sum_{i \in N} J_i (s) q_i(s) \\
    & = \sum_{i \in N} \ [\sum_{j \in N} c_j s_j - c_i \frac{1-F(s_i)}{f(s_i)}] q_i(s) \\
    & = \ \sum_{j \in N} s_j c_j \sum_{i \in N} q_i(s) - \sum_{i \in N} c_i [\frac{1-F(s_i)}{f(s_i)} q_i(s)]
\end{align*}

The first term above is unchanged given a fixed total probability of allocation. Again pointwise maximizing, we can allocate with probability 1 to the agent with the lowest adjusted hazard rate and maximize the entire term at a given signal profile.

Since the hazard rate is monotone non-increasing, once an agent starts winning by having the lowest adjusted hazard rate, she continues to win at higher signals. Thus this pointwise-maximizing allocation rule also turns out to be non-decreasing, and hence the maximizer within the class of non-decreasing allocation rules.
\end{proof}

\begin{proof}[\textbf{Proof of Proposition \ref{maxthm}}]
We invoke the optimality of the inclusive posted price under must-sell shown by Bergemann Brooks Morris (2019). For this we first show that ex post incentive compatibility is sufficient for interim incentive compatibility.

EPIC means $\forall i$, $\forall s_{-i}$, $\forall s_i, s_i'$,

$$u_i(s_i,s_{-i})\geq u_i(s_i',s_{-i})+[v_i(s_i,s_{-i})-v_i(s_i',s_{-i})]q_i(s_i',s_{-i})$$

$$\Rightarrow \int_{s_{-i}}u_i(s_i,s_{-i}) dG(s_{-i}) \geq \int_{s_{-i}} u_i(s_i',s_{-i})+[v_i(s_i,s_{-i})-v_i(s_i',s_{-i})]q_i(s_i',s_{-i}) \ dG(s_{-i}) \ \ \ \  \forall s_i,s_i', \ \forall i$$

Now, we show that ex post individual rationality is sufficient for interim individual rationality.

EPIR means $\forall i$, $\forall s_{-i}$, $\forall s_i$,
$$u_i(s_i,s_{-i}) \geq 0$$

$$\Rightarrow \int_{s_{-i}} u_i(s_i,s_{-i}) \ dG(s_{-i}) \geq 0 \ \ \ \  \forall s_i, \ \forall i$$

Since the allocation rule $q_i(s)=\frac{1}{N}, \ \forall i$, is monotone and hence eventually monotone, it is ex-post implementable. The corresponding payments are,

For agent $j$ such that $s_j \geq s_i, \forall i$

$$p_j(s)=\frac{1}{N}s_j-\int_{\mathrm{max} \ s_{-j}}^{s_j}  \frac{1}{N} dx$$
$$=\frac{\mathrm{max} \ s_{-j}}{N}$$

and for all $i \neq j$
$$p_i(s)=\frac{s_j}{N}$$

Thus the tuple $(q,p)$ is both EPIC and EPIR, and hence IIR and IIC.

We now show that in expectation, this payment scheme gives the same revenue as the inclusive posted price mechansim in Bergemann Brooks Morris (2019).

The expression for expected revenue is
$$\sum_{i \in N} \int_{s \in S}p_i(s)f(s) \ ds $$

Changing the order of summation,

$$=\int_{s \in S} \left\{ \sum_{i \in N} p_i(s) \right \} f(s) ds$$
$$=\int_{s \in S} \left\{ \frac{N-1}{N} s_{(N)} + \frac{1}{N} s_{(N-1)} \right \} f(s) ds\footnote{$s_{(j)}$ is the j-th order statistic}$$

$$=\frac{N-1}{N} E(s_{(N)}) + \frac{1}{N} E(s_{(N-1)})$$

$$=\frac{(N-1)}{N}\int_{s \in S} x N (F(x))^{N-1} f(x) dx + \frac{1}{N}\int_{s \in S} x N {N-1 \choose N-2} (1-F(x))(F(x))^{N-2}f(x)dx$$

$$=(N-1)\int_{s \in S} x (F(x))^{N-1} f(x) dx + (N-1) \int_{s \in S} x (1-F(x))(F(x))^{N-2}f(x)dx$$

$$=(N-1) \int_{s \in S} x (F(x))^{N-2} f(x) dx$$

$$=\int_{\underline{s}}^{\overline{s}} x \ d(F^{N-1}(x)) $$

It is direct to show that any mechanism in the class of mechanisms identified in this theorem has the same revenue as the equal sharing mechanism, and hence we are done.

\end{proof}

\end{appendices}

\end{spacing}
\end{document}